\documentclass{sig-alternate-2013}
\usepackage{epsfig,endnotes}
\usepackage{ifpdf}
\usepackage{multirow}
\usepackage{color}
\usepackage{etoolbox}
\newcommand{\subparagraph}{}
\usepackage{titlesec}

\usepackage{amsmath, hyperref}
\usepackage[numbers]{natbib}

\newtheorem{theorem}{Theorem}
\newtheorem{lemma}{Lemma}
\newtheorem{obs}{Observation}
\newcommand\RAPPOR{{RAPPOR}}
\newcommand\RAPPORSPELLEDOUT{{Randomized Aggregatable Privacy-Preserving Ordinal Response}}

\newcommand\USC{{University of Southern California}}

\title{\RAPPOR{}: \RAPPORSPELLEDOUT{}}
\author{{\'{U}}lfar Erlingsson, Aleksandra Korolova and Vasyl Pihur}

\numberofauthors{3} 
\author{
\alignauthor
{\'{U}l}far Erlingsson \\
       \affaddr{Google, Inc.} \\
       \email{ulfar@google.com}
\alignauthor
Vasyl Pihur \\
       \affaddr{Google, Inc.}\\
       \email{vpihur@google.com}
\alignauthor Aleksandra Korolova \\[0.2ex]
       \affaddr{\makebox[1em][c]{\USC}}\\[0.2ex]
       \email{korolova@usc.edu}
}

\newfont{\mycrnotice}{ptmr8t at 7pt}
\newfont{\myconfname}{ptmri8t at 7pt}
\let\crnotice\mycrnotice%
\let\confname\myconfname%

\permission{Permission to make digital or hard copies of part or all of this work for personal or classroom use is granted without fee provided that copies are not made or distributed for profit or commercial advantage, and that copies bear this notice and the full citation on the first page. Copyrights for third-party components of this work must be honored. For all other uses, contact the owner/authors. Copyright is held by the authors.}
\conferenceinfo{CCS'14,}{November 3--7, 2014, Scottsdale, Arizona, USA.} 
\copyrightetc{ACM 978-1-4503-2957-6/14/11, http://dx.doi.org/10.1145/2660267.2660348}

\begin{document}
\maketitle

\begin{abstract}
\RAPPORSPELLEDOUT{},
or \RAPPOR{}, is a technology for crowdsourcing statistics from end-user client software, 
anonymously, with strong privacy guarantees.  
In short, \RAPPOR{}s allow the forest of client data to be studied, 
without permitting the possibility of looking at individual trees.
By applying randomized response in a novel manner,
\RAPPOR{} provides the mechanisms for such collection as well as for efficient, 
high-utility analysis of the collected data.
In particular, \RAPPOR{} permits statistics to be collected
on the population of client-side strings
with strong privacy guarantees for each client,
and without linkability of their reports.

This paper
describes and motivates \RAPPOR{},
details its differential-privacy and utility guarantees,
discusses its practical deployment and properties in the face of different attack models,
and, finally, gives results of its application
to both synthetic and real-world data.
\end{abstract}

\section{Introduction}
Crowdsourcing data to make better, more informed decisions is becoming increasingly commonplace. 
For any such crowdsourcing,
privacy-preservation mechanisms should be applied 
to reduce and control the privacy risks introduced by the data collection process,
and balance that risk against the beneficial utility of the collected data.
For this purpose we introduce
\RAPPORSPELLEDOUT{},
or \RAPPOR{},
a widely-applicable, practical new mechanism that provides strong privacy guarantees 
combined with high utility,
yet is not founded on the use of trusted third parties.

\RAPPOR{} builds on the ideas of \textit{randomized response}, a surveying technique developed in the 1960s for collecting statistics on sensitive topics where survey respondents wish to retain confidentiality~\citep{warner}.
An example commonly used to describe this technique involves a question on a sensitive topic, such as ``Are you a member of the Communist party?''~\cite{WikipediaRR}. 
For this question, the survey respondent is asked to flip a fair coin, in secret, and answer ``Yes'' if it comes up heads, but tell the truth otherwise (if the coin comes up tails).
Using this procedure, each respondent retains very strong deniability for any ``Yes'' answers, since such answers are most likely attributable to the coin coming up heads; as a refinement, respondents can also choose the untruthful answer by flipping another coin in secret, and get strong deniability for both ``Yes'' and ``No'' answers.

Surveys relying on randomized response enable easy computations of accurate population statistics while preserving the privacy of the individuals. Assuming absolute compliance with the randomization protocol (an assumption that may not hold for human subjects, and can even be non-trivial for algorithmic implementations~\cite{mironov-CCS12}), it is easy to see that in a case where both ``Yes'' and ``No'' answers can be denied (flipping two fair coins), the true number of ``Yes'' answers can be accurately estimated by $2(Y - 0.25)$, where $Y$ is the proportion of ``Yes'' responses. 
%
%
In expectation,
respondents will provide the true answer 75\% of the time,
as is easy to see by
a case analysis of the two fair coin flips.

Importantly,
for one-time collection,
the above randomized survey mechanism
will protect the privacy
of any specific respondent,
irrespective of any attacker's prior knowledge, 
as assessed via the $\epsilon$-differential privacy guarantee~\cite{dwork06}.
Specifically,
the respondents will have
differential privacy at the level $\epsilon = \ln\bigl(0.75 / (1 - 0.75)\bigr) = \ln(3)$.
This said,
this privacy guarantee degrades
if the survey is repeated---e.g.,
to get fresh, daily statistics---and
data is collected multiple times from the same respondent.
In this case,
to maintain both differential privacy and utility,
better mechanisms are needed,
like those we present in this paper.

Privacy-Preserving Aggregatable Randomized Response, 
or \RAPPOR{}s, is a new mechanism for collecting statistics from end-user, client-side software, 
in a manner that provides strong privacy protection
using randomized response techniques.  
\RAPPOR{} is designed to permit collecting, over large numbers of clients, 
statistics on client-side values and strings, 
such as their categories, frequencies, histograms, 
and other set statistics.
For any given value reported,
\RAPPOR{}
gives 
a strong deniability guarantee for the reporting client,
which strictly limits private information disclosed,
as measured by an $\epsilon$-differential privacy bound,
and holds even for 
a single client that reports often on the same value.

%
A distinct contribution is \RAPPOR{}'s ability 
to collect statistics about an arbitrary set of strings 
by applying randomized response to Bloom filters~\cite{bloom} 
with strong $\epsilon$-differential privacy guarantees. 
Another contribution is the elegant manner
in which \RAPPOR{}
protects the privacy of clients 
from whom data is collected repeatedly (or even infinitely often), and
how \RAPPOR{}
avoids addition of privacy externalities,
such as those that might be created by
maintaining a database of contributing respondents (which might be breached),
or repeating a single, memoized response (which would be linkable, and might be tracked).
In comparison,
traditional randomized response does not provide any longitudinal privacy 
in the case when multiple responses are collected from the same participant. 
Yet another contribution is that the \RAPPOR{}  mechanism is performed locally on the client, and does not require a trusted third party.

Finally, \RAPPOR{}
provides a novel, high-utility 
decoding framework for learning statistics
based on a sophisticated combination of
hypotheses testing,
least-squares solving, and LASSO regression~\cite{lasso}.

\subsection{The Motivating Application Domain}\label{sec:motivation}
\RAPPOR{} is
a general technology for privacy-preserving data collection and crowdsourcing of statistics,
which could be applied in a broad range of contexts.

In this paper, however, we focus on the specific application domain
that motivated the development of \RAPPOR{}:
the need for Cloud service operators to collect up-to-date statistics 
about the activity of their users and their client-side software.
In this domain,
\RAPPOR{} has already seen limited deployment in Google's Chrome Web browser,
where it has been used
to improve the data sent by users that have opted-in to reporting statistics~\cite{ChromeRAPPORpage}.
Section~\ref{sec:chromehome} briefly describes this real-world application,
and the benefits \RAPPOR{} has provided
by shining a light on the unwanted or malicious hijacking of user settings.

For a variety of reasons,
understanding population statistics is a key part  
of an effective, reliable operation of online services
by Cloud service and software platform operators.
These reasons are often as simple as 
observing how frequently certain software features are used,
and measuring their performance and failure characteristics.
Another, important set of reasons
involve
providing better security and abuse protection to the users, their clients, and the service itself.
For example, to assess the prevalence of botnets or hijacked clients,
an operator may wish to monitor how many clients
have---in the last 24 hours---had critical preferences overridden, 
e.g., to redirect
the users' Web searches to the URL of a known-to-be-malicious search provider.
%

The collection of up-to-date crowdsourced statistics
raises a dilemma for service operators.
On one hand,
it will likely be detrimental to the end-users' privacy
to directly collect their information.
(Note that even the search-provider preferences of a user
may be uniquely identifying, incriminating, 
or otherwise compromising for that user.)
On the other hand,
not collecting any such information 
will also be to the users' detriment:
if operators cannot
gather the right statistics,
they cannot make
many software and service improvements that benefit users
(e.g., by detecting or preventing malicious client-side activity).
Typically, operators resolve this dilemma
by using techniques
that derive only the necessary high-order statistics,
using mechanisms
that limit the users' privacy risks---for example,
by collecting only coarse-granularity data, 
and by eliding data that is not shared by a certain number of users.

Unfortunately,
even for careful operators, 
willing to utilize state-of-the-art techniques,
there are few existing, practical mechanisms 
that offer both privacy and utility,
and even fewer that provide clear privacy-protection guarantees.
%
%
Therefore, 
to reduce privacy risks,
operators rely to a great extent on pragmatic means and processes,
that, for example, avoid the collection of data, 
remove unique identifiers, or otherwise systematically scrub data, 
perform mandatory deletion of data after a certain time period, 
and, in general, enforce access-control and auditing policies on data use.
However, these approaches are limited in their ability to provide provably-strong privacy guarantees. 
In addition, privacy externalities from individual data collections, 
such as timestamps or linkable identifiers, may arise;
the privacy impact of those externalities may
be even greater than that of the data collected.

\RAPPOR{} can help operators
handle the significant challenges, and potential privacy pitfalls,
raised by this dilemma.

\subsection{Crowdsourcing Statistics with \RAPPOR{}}
Service operators may apply \RAPPOR{}
to crowdsource statistics
in a manner that protects their users' privacy,
and thus address the challenges described above.

As a simplification,
\RAPPOR{} responses can be assumed to be \emph{bit strings},
where each bit corresponds to a randomized response 
for some logical predicate on the reporting client's properties, such as its values, context, or history.
(Without loss of generality,
this assumption
is used for the remainder of this paper.)
For example, one bit in a \RAPPOR{} response may correspond to a predicate 
that indicates the stated gender, male or female, of the client user,
or---just as well---their membership in the Communist party.

The structure of a \RAPPOR{} response need not be otherwise constrained;
in particular, 
(i) the response bits may be sequential, or unordered,
(ii) the response predicates may be independent, disjoint, or correlated,
and (iii) the client's properties may be immutable, or changing over time.
However,
those details (e.g., any correlation of the response bits)
must be correctly accounted for,
as they impact both the utilization and privacy guarantees of \RAPPOR{}---as
outlined in the next section, and detailed in later sections.

In particular,
\RAPPOR{} can be used to
collect 
statistics on 
categorical client properties,
by having each bit in a client's response
represent whether, or not, that client belongs to a category.
For example, those categorical predicates 
might represent 
whether, or not, the client is utilizing a software feature.
In this case,
if each client can use only one of three disjoint features, $X, Y$, and $Z$, 
the collection of a three-bit \RAPPOR{} response from clients
will allow measuring
the relative frequency 
by which the features are used by clients.
As regards to privacy,
each client will be protected
by the manner in which 
the three bits are derived from a single (at most) true predicate;
as regards to utility,
it will suffice to count how many responses had the bit set, for each distinct response bit,
to get a good statistical estimate of the empirical distribution of the features' use.

\RAPPOR{}
can also be used to 
gather
population statistics on 
numerical and ordinal values,
e.g., 
by associating response bits
with predicates for different ranges of numerical values,
or by reporting on
disjoint categories for different logarithmic magnitudes of the values.
For such numerical \RAPPOR{} statistics,
the estimate may be improved
by collecting and utilizing
relevant information about the priors and shape of
the empirical distribution, such as its smoothness.

Finally,
\RAPPOR{}
also allows collecting statistics on
non-categorical domains, or categories that cannot be enumerated ahead of time,
through the use of Bloom filters~\cite{bloom}.
In particular, \RAPPOR{}
allows collection of compact Bloom-filter-based randomized responses on strings,
instead of 
having clients report
when they match a set of hand-picked strings, predefined by the operator.
Subsequently,
those responses can be matched against candidate strings, as they become known to the operator,
and used to estimate both known and unknown strings in the population.
Advanced statistical decoding techniques must be applied
to accurately interpret the randomized, noisy data in Bloom-filter-based \RAPPOR{} responses.
However, as in the case of categories,
this analysis needs only consider the aggregate counts of distinct bits set in \RAPPOR{} responses
to provide good estimators for population statistics,
as detailed in Section~\ref{sec:decoding}.

Without loss of privacy,
\RAPPOR{} analysis can be re-run on a collection of responses, e.g., 
to consider new strings and cases missed in previous analyses,
without the need to re-run the data collection step.
Individual responses
can be especially useful for exploratory or custom data analyses. 
For example, if the geolocation of clients' IP addresses are collected
alongside the \RAPPOR{} reports of their sensitive values, 
then the observed distributions
of those values 
could be compared across different geolocations,
e.g., by analyzing different subsets separately.
Such analysis
is compatible with \RAPPOR{}'s privacy guarantees,
which hold true even in the presence of auxiliary data,
such as geolocation.
By limiting the number of correlated categories, or Bloom filter hash functions,
reported by any single client,
\RAPPOR{} can maintain its differential-privacy guarantees even when 
statistics are collected on multiple aspects of clients,
as outlined next, and detailed in Sections~\ref{sec:diffprivacy} and~\ref{sec:attacks}.



\subsection{\RAPPOR{} and (Longitudinal) Attacks}
Protecting privacy for both one-time and multiple collections requires consideration of several distinct attack models.
 A basic attacker is assumed to have access to a single report and can be stopped with a single round of randomized response.
A windowed attacker has access to multiple reports over time from the same user.
Without careful modification of the traditional randomized response techniques, almost certainly full disclosure of private information would happen.
This is especially true if the window of observation is large and the underlying value does not change much.
An attacker with complete access to all clients' reports (for example, an insider with unlimited access rights),
is the hardest to stop, yet such an attack is also the most difficult to execute in practice. \RAPPOR{} provides explicit trade-offs between
different attack models in terms of tunable privacy protection for all three types of attackers.

\RAPPOR{} builds on the basic idea of memoization and provides a framework for one-time and longitudinal privacy protection by playing the randomized response game twice with a memoization step in between. The first step, called a Permanent randomized response, is used to create a ``noisy'' answer which is memoized by the client and permanently reused in place of the real answer. The second step, called an Instantaneous randomized response, reports on the ``noisy'' answer over time, eventually completely revealing it. 
Long-term, longitudinal privacy is ensured by the use of the Permanent randomized response, while the use of an Instantaneous randomized response provides protection against possible tracking externalities.

The idea of \emph{underlying memoization} turns out to be crucial for privacy protection in the case where multiple responses are collected from the same participant over time. For example, in the case of the question about the Communist party from the start of the paper, memoization can allow us to provide $\ln(3)$-differential privacy even with an \emph{infinite} number of responses, as long as the underlying memoized response has that level of differential privacy.

On the other hand,
without memoization or other limitation on responses,
randomization is not sufficient to
maintain plausible deniability
in the face of multiple collections.
For example,
if 75 out of 100 responses are ``Yes'' for a single client
in the randomized-response scheme
at the very start of this paper,
the true answer will have been ``No''
in a vanishingly unlikely $1.39 \times 10^{-24}$
fraction of cases.

Memoization is absolutely effective in providing longitudinal privacy only in cases when the underlying true value does not change
 or changes in an uncorrelated fashion. When users' consecutive reports are temporally correlated, differential privacy guarantees
 deviate from their nominal levels and become progressively weaker as correlations increase. Taken to the extreme, when asking users
 to report daily on their age in days, additional measures are required to prevent full disclosure over time, such as stopping collection after
 a certain number of reports or increasing the noise levels exponentially, 
as discussed further in Section~\ref{sec:attacks}.
 
For a client that reports on a property that strictly alternates between two true values, ($a, b, a, b, a, b, a, b, \ldots$), the two memoized Permanent randomized
 responses for $a$ and $b$ will be reused, again and again, to generate \RAPPOR{} report data. Thus, an attacker that obtains a large enough number of reports,
 could learn those memoized ``noisy" values with arbitrary certainty---e.g.,
 by separately analyzing the even and odd subsequences.
However, even in this case, 
the attacker cannot be certain of the values of $a$ and $b$
because of memoization.
This said,
if $a$ and $b$ are correlated,
the attacker may still learn more
than they otherwise would have;
maintaining privacy in the face of 
any such correlation 
is discussed further in Sections~\ref{sec:diffprivacy} and~\ref{sec:attacks} (see also~\cite{KiferM11}).

\begin{figure*}[!t]
\centering
\includegraphics[trim=0 0.2in 0 0,clip=true,width=2\columnwidth]{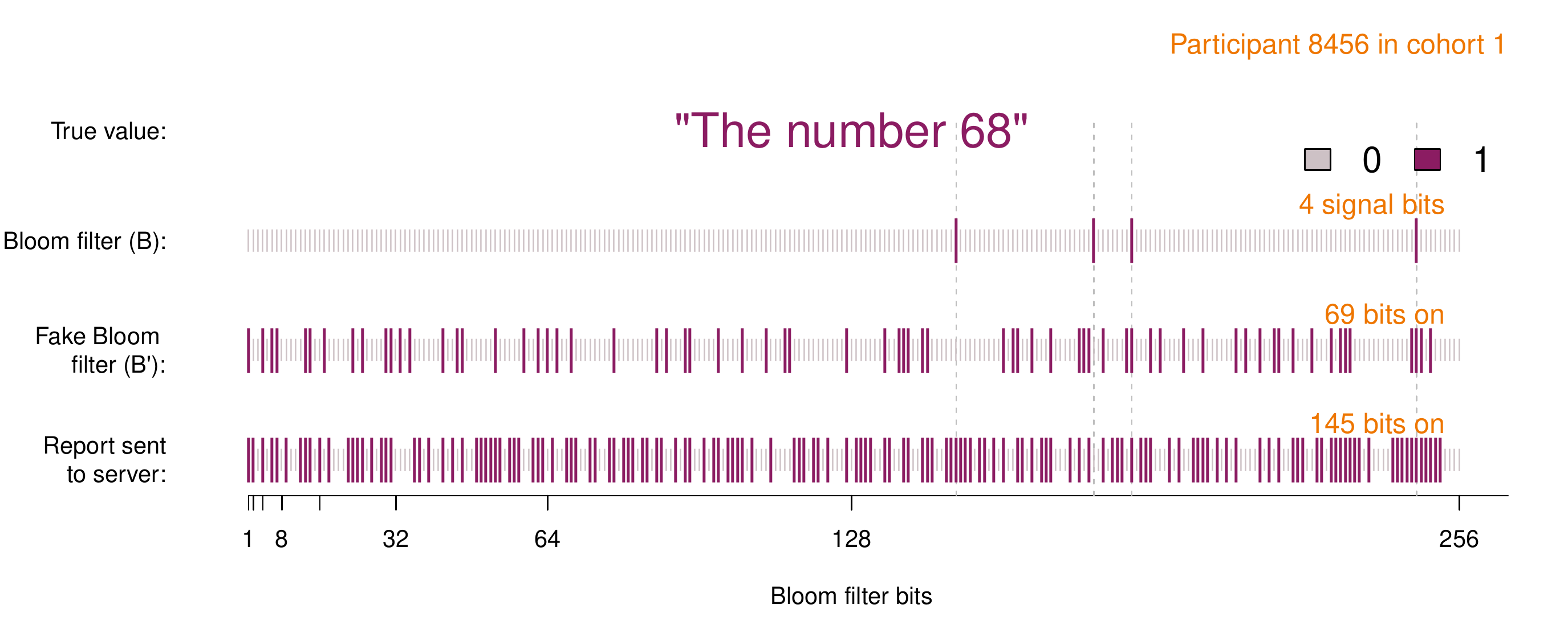}
\caption{Life of a \RAPPOR{} report: The client value of the string ``The number 68'' is hashed onto the Bloom filter $B$ using $h$ (here 4) hash functions. For this string, a Permanent randomized response $B'$ is produces and memoized by the client, and this $B'$ is used (and reused in the future) to generate Instantaneous randomized responses $S$ (the bottom row), which are sent to the collecting service.}\vspace*{-2ex}
\label{fig:life}
\end{figure*}

In the next section we will describe the \RAPPOR{} algorithm in detail. We then provide intuition and formal justification for the reasons why the proposed algorithm satisfies the rigorous privacy guarantees of differential privacy. We then devote several sections to discussion of the additional technical aspects of \RAPPOR{} that are crucial for its potential uses in practice, such as parameter selection, interpretation of results via advanced statistical decoding, and experiments illustrating what can be learned in practice. The remaining sections discuss our experimental evaluation, the attack models we consider, the limitations of the \RAPPOR{} technique, as well as related work.

\section{The Fundamental \RAPPOR{} Algorithm}
Given a client's value $v$, the \RAPPOR{} algorithm executed by the client's machine, reports to the server a bit array of size $k$, that encodes a ``noisy" representation of its true value $v$. The noisy representation of $v$ is chosen in such a way so as to reveal a \emph{controlled} amount of information about $v$, limiting the server's ability to learn with confidence what $v$ was. This remains true even for a client that submits an infinite number of reports on a particular value $v$.

To provide such strong privacy guarantees, the \RAPPOR{} algorithm implements two separate defense mechanisms, both of which are based on the idea of randomized response and can be separately tuned depending on the desired level of privacy protection at each level. 
Furthermore, additional uncertainty is added through the use of Bloom filters which serve not only to make reports compact, but also to complicate the life of any attacker (since any one bit in the Bloom filter may have multiple data items in its pre-image).

The \RAPPOR{} algorithm takes in the client's true value $v$ and parameters of execution $k, h, f, p, q$, and is executed locally on the client's machine performing the following steps:

\begin{enumerate}
\item {\bf Signal.} Hash client's value $v$ onto the Bloom filter $B$ of size $k$ using $h$ hash functions.
\item {\bf Permanent randomized response.} For each client's value $v$ and bit $i, 0 \leq i < k$ in $B$, create a binary reporting value $B'_i$ which equals to
$$
B'_i = \begin{cases}
1, & \text{with probability $\frac{1}{2}f$} \\
0, & \text{with probability $\frac{1}{2}f$} \\
B_i, & \text{with probability $1 - f$}
\end{cases}
$$
where $f$ is a user-tunable parameter controlling the level of longitudinal privacy guarantee.\\[1ex]
Subsequently, this $B'$ is memoized and reused as the basis for all future reports on this distinct value $v$.
\item {\bf Instantaneous randomized response.} Allocate a bit array $S$ of size $k$ and initialize to 0. Set each bit $i$ in $S$ with probabilities
$$
P(S_i = 1) = \begin{cases}
q, & \text{if $B'_i = 1$}. \\
p, & \text{if $B'_i = 0$}.
\end{cases}
$$
\item {\bf Report.} Send the generated report $S$ to the server.
\end{enumerate}

There are many different variants of the above randomized response mechanism. Our main objective for selecting these two particular versions was to make the scheme intuitive and easy to explain. 

The Permanent randomized response (step 2) replaces the real value $B$ with a derived randomized noisy value $B'$. $B'$ may or may not contain any information about $B$ depending on whether signal bits from the Bloom filter are being replaced by random 0's with probability $\frac{1}{2}f$. The Permanent randomized response ensures privacy because of the adversary's limited ability to differentiate between true and ``noisy'' signal bits. It is absolutely critical that all future reporting on the information about $B$ uses the same randomized $B'$ value to avoid an ``averaging" attack, in which an adversary estimates the true value from observing multiple noisy versions of it. 

The Instantaneous randomized response (step 3) plays several important functions. Instead of directly reporting $B'$ on every request, the client reports a randomized version of $B'$. This modification significantly increases the difficulty of tracking a client based on $B'$, which could otherwise be viewed as a unique identifier in longitudinal reporting scenarios. It also provides stronger short-term privacy guarantees (since we are adding more noise to the report) which can be independently tuned to balance short-term vs long-term risks. Through tuning of the parameters of this mechanism we can effectively balance utility against different attacker models.

Figure~\ref{fig:life} shows a random run of the \RAPPOR{} algorithm. Here, a client's value is $v = ``68"$, the size of the Bloom filter is $k = 256$, the number of hash functions is $h = 4$, and the tunable randomized response parameters are: $p = 0.5$, $q = 0.75$, and $f = 0.5$. The reported bit array sent to the server is shown at the bottom of the figure. 145 out of 256 bits are set in the report. Of the four Bloom filter bits in $B$ (second row), two are propagated to the noisy Bloom filter $B'$. Of these two bits, both are turned on in the final report. The other two bits are never reported on by this client due to the permanent nature of $B'$. With multiple collections from this client on the value ``68'', the most powerful attacker would eventually learn $B'$ but would continue to have limited ability to reason about the value of $B$, as measured by differential privacy guarantee. In practice, learning about the actual client's value $v$ is even harder because multiple values map to the same bits in the Bloom filter \citep{bloom_privacy}.

\subsection{\RAPPOR{} Modifications}
The \RAPPOR{} algorithm can be modified in a number of ways depending on the particulars of the scenario in which privacy-preserving data collection is needed. Here, we list three common scenarios where omitting certain elements from the \RAPPOR{} algorithm leads to a more efficient learning procedure, especially with smaller sample sizes. 
\begin{itemize}
\item {\bf One-time \RAPPOR{}.} One time collection, enforced by the client, does not require longitudinal privacy protection. The Instantaneous randomized response step can be skipped in this case and a direct randomization on the true client's value is sufficient to provide strong privacy protection. 
\item {\bf Basic \RAPPOR{}.} If the set of strings being collected is relatively small and well-defined, such that each string can be deterministically mapped to a single bit in the bit array, there is no need for using a Bloom filter with multiple hash functions. For example, collecting data on client's gender could simply use a two-bit array with ``male'' mapped to bit 1 and ``female'' mapped to bit 2. This modification would affect step 1, where a Bloom filter would be replaced by a deterministic mapping of each candidate string to one and only one bit in the bit array. In this case, the effective number of hash functions, $h$, would be 1.
\item {\bf Basic One-time \RAPPOR{}.} This is the simplest configuration of the \RAPPOR{} mechanism, combining the first two modifications at the same time: one round of randomization using a deterministic mapping of strings into their own unique bits.
\end{itemize}

\section{Differential Privacy of \RAPPOR{}}\label{sec:diffprivacy}
The scale and availability of data in today's world makes increasingly sophisticated attacks feasible, and any system that hopes to withstand such attacks should aim to ensure rigorous, rather than merely intuitive privacy guarantees. For our analysis, we adopt the rigorous notion of privacy, \textit{differential privacy}, which was introduced by Dwork et al~\citep{dwork06} and has been widely adopted~\cite{DworkCACM}. The definition aims to ensure that the output of the algorithm does not significantly depend on any particular individual's data. The quantification of the increased risk that participation in a service poses to an individual can, therefore, empower clients to make a better informed decision as to whether they want their data to be part of the collection.

Formally, a randomized algorithm $A$ satisfies $\epsilon$-differential privacy~\citep{dwork06} if for all pairs of client's values $v_1$ and $v_2$ and for all $R\subseteq Range(A)$,
$$
P(A(v_1) \in R) \le e^\epsilon P(A(v_2) \in R).
$$

We prove that the \RAPPOR{} algorithm satisfies the definition of differential privacy next. 
Intuitively, the Permanent randomized response part ensures that the ``noisy" value derived from the true value protects privacy, and the Instantaneous randomized response provides protection against usage of that response by a longitudinal tracker.




\subsection{Differential Privacy of the Permanent Randomized Response}
\begin{theorem}\label{thm-one}
The Permanent randomized response (Steps 1 and 2 of \RAPPOR{}) satisfies $\epsilon_{\infty}$-differential privacy where $\epsilon_{\infty} = 2h\ln\left(\frac{1 - \frac{1}{2}f}{\frac{1}{2}f}\right)$.
\end{theorem}

\begin{proof}
Let $S = s_1, \ldots, s_k$ be a randomized report generated by the \RAPPOR{} algorithm. Then the probability of observing any given report $S$ given the true client value $v$ and assuming that $B'$ is known is
\begin{eqnarray*}
P(S = s | V = v) & = & P(S = s | B, B', v)\cdot P(B' | B, v)\cdot P(B | v) \\
                         & = & P(S = s | B')\cdot P(B' | B)\cdot P(B | v) \\
                         & = & P(S = s | B')\cdot P(B' | B).
\end{eqnarray*}

Because $S$ is conditionally independent of $B$ given $B'$, the first probability provides no additional information about $B$. 
$P(B' | B)$ is, however, critical for longitudinal privacy protection. Relevant probabilities are
\begin{eqnarray*}
P(b'_i = 1 | b_i = 1) & = & \frac{1}{2}f + 1 - f = 1 - \frac{1}{2}f \;\;\;\text{~and~}\\
P(b'_i = 1 | b_i = 0) & = & \frac{1}{2}f.
\end{eqnarray*}

Without loss of generality, let the Bloom filter bits $1, \ldots, h$ be set, i.e., $b^* = \{b_1 = 1, \ldots, b_h = 1, b_{h+1} = 0, \ldots, b_{k} = 0\}$. Then,
\begin{eqnarray*}
P(B' = b' | B = b^*) & = & \left(\frac{1}{2}f\right)^{b'_1}\left(1 - \frac{1}{2}f\right)^{1 - b'_1} \times \ldots \\
                   &   & \times \left(\frac{1}{2}f\right)^{b'_h}\left(1 - \frac{1}{2}f\right)^{1 - b'_h} \times \ldots \\
                   &   & \times \left(1 - \frac{1}{2}f\right)^{b'_{h + 1}}\left(\frac{1}{2}f\right)^{1 - b'_{h + 1}} \times \ldots \\
                   &   & \times \left(1 - \frac{1}{2}f\right)^{b'_k}\left(\frac{1}{2}f\right)^{1 - b'_k}.
\end{eqnarray*}

Let $RR_{\infty}$ be the ratio of two such conditional probabilities with distinct values of $B$, $B_1$ and $B_2$, i.e., $RR_{\infty} = \frac{P(B'\in R^* | B = B_1)}{P(B' \in R^* | B = B_2)}$. For the differential privacy condition to hold, $RR_{\infty}$ needs to be bounded by $\exp(\epsilon_{\infty})$.

\begin{eqnarray*}
RR_{\infty} & = &  \frac{P(B'\in R^* | B = B_1)}{P(B' \in R^* | B = B_2)} \\
                  & =  & \frac{\sum_{B'_i \in R^*}P(B' = B'_i |B = B_1)}{\sum_{B'_i \in R^*} P(B' = B'_i | B = B_2)} \\
                  & \le & \max_{B'_i \in R^*} \frac{P(B' = B'_i | B = B_1)}{P(B' = B'_i| B = B_2)}  \;\;\;\;\; \text{(by Observation~\ref{obs-ratios})} \\
                 & =  & \left(\frac{1}{2}f\right)^{2(b'_1 + b'_2 + \ldots + b'_h - b'_{h+1} - b'_{h+2} - \ldots - b'_{2h})} \\
                  & & \times \left(1 - \frac{1}{2}f\right)^{2(b'_{h+1} + b'_{h+2} + \ldots + b'_{2h} - b'_1 - b'_2 - \ldots - b'_h)}.
\end{eqnarray*}

Sensitivity is maximized when $b'_{h+1} = b'_{h+2} = \ldots = b'_{2h} = 1$ and $b'_1 = b'_2 = \ldots = b'_h = 0$. Then,\\
$
RR_{\infty} = \left(\frac{1 - \frac{1}{2}f}{\frac{1}{2}f}\right)^{2h} \text{~and~}
\epsilon_{\infty} = 2h\ln\left(\frac{1 - \frac{1}{2}f}{\frac{1}{2}f}\right).
$~\end{proof}

Note that $\epsilon_{\infty}$ is not a function of $k$. It is true that a smaller $k$, or a higher rate of Bloom filter bit collision, sometimes improves privacy protection, but, on its own, it is not sufficient nor necessary to provide $\epsilon$-differential privacy.

\subsection{Differential Privacy of the Instantaneous Randomized Response}
With a single data collection from each client, the attacker's knowledge of $B$ must come directly from a single report $S$ generated by applying the randomization twice, thus, providing a higher level of privacy protection than under the assumption of complete knowledge of $B'$.

Because of a two-step randomization, probability of observing a 1 in a report is a function of both $q$ and $p$ as well as $f$.
\begin{lemma}\label{lem-one}
Probability of observing 1 given that the underlying Bloom filter bit was set is given by
$$
q^* = P(S_i = 1 | b_i = 1) = \frac{1}{2}f(p + q) + (1 - f)q.
$$
Probability of observing 1 given that the underlying Bloom filter bit was \emph{not} set is given by
$$
p^* = P(S_i = 1 | b_i = 0) = \frac{1}{2}f(p + q) + (1 - f)p.
$$
\end{lemma}
We omit the proof as the reasoning is straightforward that probabilities in both cases are mixtures of random and true responses with the mixing proportion $f$.

\begin{theorem}
The Instantaneous randomized response (Step 3 of \RAPPOR{}) satisfies $\epsilon_1$-differential privacy, where $\epsilon_1 = h\log\left(\frac{q^*(1 - p^*)}{p^*(1 - q^*)}\right)$ and $q^*$ and $p^*$ as defined in Lemma~\ref{lem-one}.
\end{theorem}

\begin{proof}
The proof is analogous to Theorem~\ref{thm-one}. Let $RR_1$ be the ratio of two conditional probabilities, i.e., $RR_1 = \frac{P(S \in R | B = B_1)}{P(S \in R | B = B_2)}$. To satisfy the differential privacy condition, this ratio must be bounded by $\exp(\epsilon_1)$.
\begin{eqnarray*}
RR_1 &=& \frac{P(S \in R | B = B_1)}{P(S \in R | B = B_2)} \\
   &=& \frac{\sum_{s_j \in R}P(S = s_j | B = B_1)}{\sum_{s_j \in R}P(S = s_j | B = B_2)} \\
   &\le& \max_{s_j \in R} \frac{P(S= s_j | B = B_1)}{P(S = s_j | B = B_2)} \\
   &=& \left[\frac{q^*(1 - p^*)}{p^*(1 - q^*)}\right]^h
\end{eqnarray*}
and
$$
\epsilon_1 = h\log\left(\frac{q^*(1 - p^*)}{p^*(1 - q^*)}\right).
$$\end{proof}

The above proof naturally extends to $N$ reports, since each report that is not changed contributes a fixed amount to the total probability of observing \emph{all} reports and enters both nominator and denominator in a multiplicative way (because of independence). Since our differential privacy framework considers inputs that differ only in a single record, $j$, (reports set $D_1$ becomes $D_2$ differing in a single report $S_j$), the rest of the product terms end up canceling out in the ratio
\begin{eqnarray*}
\frac{P(S_1 = s_1, S_2 = s_2, \ldots, S_j = s_j, \ldots, S_N = s_N | B_1)}{P(S_1 = s_1, S_2 = s_2, \ldots, S_j = s_j, \ldots, S_N = s_N | B_2)}  =\\
   \frac{\prod_{i=1}^N P(S_i = s_i | B_1)}{\prod_{i=1}^N P(S_i = s_i | B_2)} = \frac{P(S_j = s_j | B_1)}{P(S_j = s_j | B_2)}.
\end{eqnarray*}

Computing $\epsilon_n$ for the $n$th collection cannot be made without additional assumptions about how effectively the attacker can learn $B'$ from the collected reports. We continue working on providing these bounds under various learning strategies. Nevertheless, as $N$ becomes large, the bound approaches $\epsilon_{\infty}$ but always remains strictly smaller.

\begin{figure*}[!t]
\centering
\includegraphics[width=2\columnwidth]{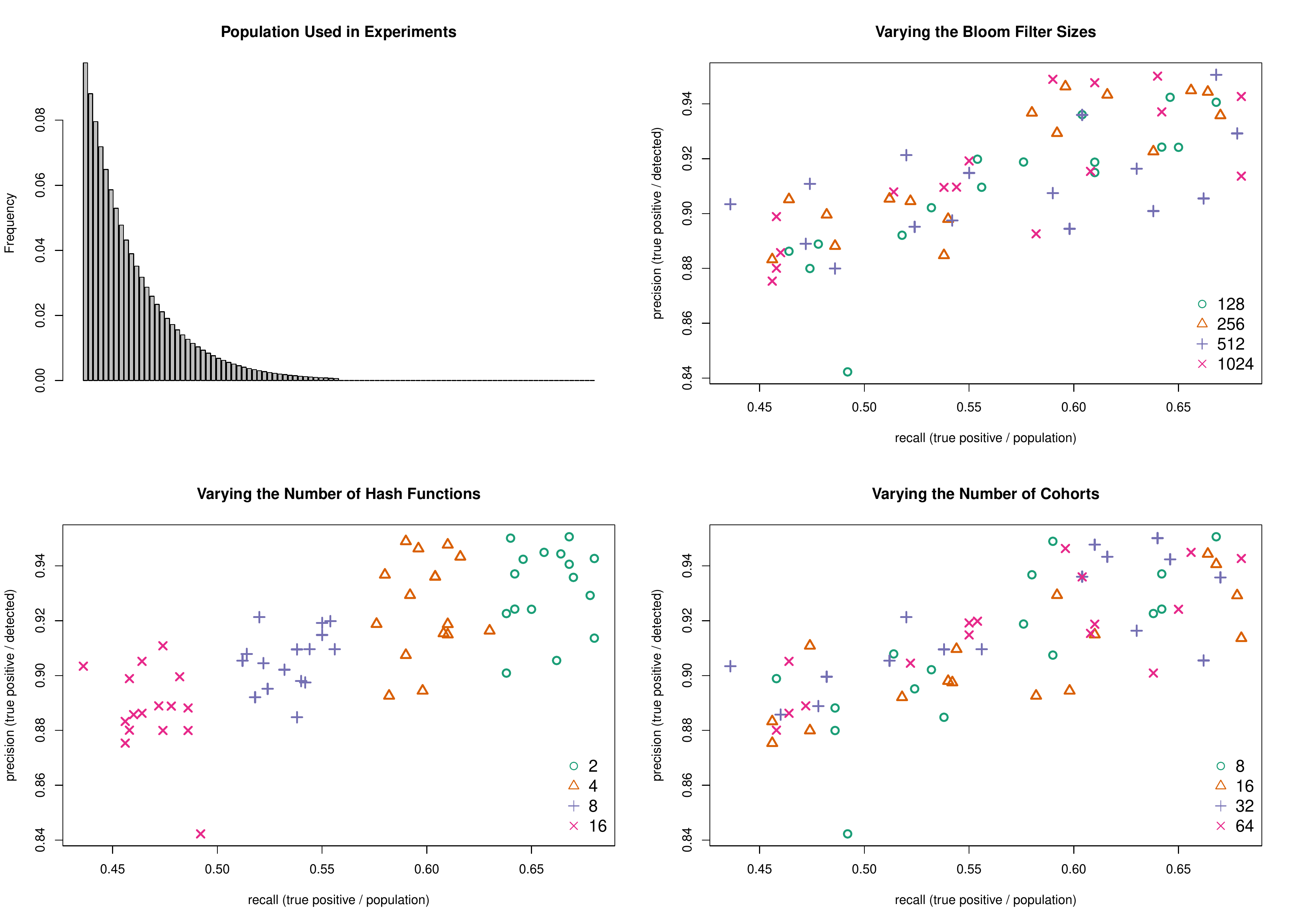}
\caption{Recall versus precision depending on choice of parameters $k$, $h$, and $m$. The first panel shows the true population distribution from which \RAPPOR{} reports were sampled. The other three panels vary one of the parameters while keeping the other two fixed. Best precision and recall are achieved with using 2 hash functions, while the choices of $k$ and $m$ do not show clear preferences.}
\label{fig:sim}
\end{figure*}

\section{High-utility Decoding of Reports}\label{sec:decoding}
In most cases, the goal of data collection using \RAPPOR{} is to learn which strings are present in the sampled population and what their corresponding frequencies are. Because we make use of the Bloom filter (loss of information) and purposefully add noise for privacy protection, decoding requires sophisticated statistical techniques.

To facilitate learning, before any data collection begins each client is randomly assigned and becomes a permanent member of one of $m$ \emph{cohorts}. Cohorts implement different sets of $h$ hash functions for their Bloom filters, thereby reducing the chance of accidental collisions of two strings across all of them. Redundancy introduced by running $m$ cohorts simultaneously greatly improves the false positive rate. The choice of $m$ should be considered carefully, however. When $m$ is too small, then collisions are still quite likely, while when $m$ is too large, then each individual cohort provides insufficient signal due to its small sample size (approximately $N/m$, where $N$ is the number of reports). Each client must report its cohort number with every submitted report, i.e., it is not private but made private.

We propose the following approach to learning from the collected reports:
\begin{itemize}
\item Estimate the number of times each bit $i$ within cohort $j$, $t_{ij}$, is truly set in $B$ for each cohort. Given the number of times each bit $i$ in cohort $j$, $c_{ij}$ was set in a set of $N_j$ reports, the estimate is given by
$$
t_{ij} = \frac{c_{ij} - (p + \frac{1}{2}fq - \frac{1}{2}fp)N_j}{(1 - f)(q - p)}.
$$
Let $Y$ be a vector of $t_{ij}$'s, $i \in [1, k], j \in[1, m]$.
\item Create a design matrix $X$ of size $km \times M$ where $M$ is the number of candidate strings under consideration. $X$ is mostly 0 (sparse) with 1's at the Bloom filter bits for each string for each cohort. So each column of $X$ contains $hm$ 1's at positions where a particular candidate string was mapped to by the Bloom filters in all $m$ cohorts. Use Lasso \citep{lasso} regression to fit a model $Y \sim X$ and select candidate strings corresponding to non-zero coefficients.
\item Fit a regular least-squares regression using the selected variables to estimate counts, their standard errors and p-values.
\item Compare p-values to a Bonferroni corrected level of $\alpha / M = 0.05 / M$ to determine which frequencies are statistically significant from 0. Alternatively, controlling the False Discovery Rate (FDR) at level $\alpha$ using the Benjamini-Hochberg procedure \cite{Benjamini1995}, for example, could be used.
\end{itemize}

\subsection{Parameter Selection}
Practical implementation of the \RAPPOR{} algorithm requires specification of a number of parameters. $p$, $q$, $f$ and the number of hash functions $h$ control the level of privacy for both one-time and longitudinal collections. Clearly, if no longitudinal data is being collected, then we can use One-time \RAPPOR{} modification. With the exception of $h$, the choice of values for these parameters should be driven exclusively by the desired level of privacy $\epsilon$. $\epsilon$ itself can be picked depending on the circumstances of the data collection process; values in the literature range from $0.01$ to $10$ (see Table 1 in \cite{HsuGHKNPR14}).

Bloom filter size, $k$, the number of cohorts, $m$, and $h$ must also be specified \emph{a priori}. Besides $h$, neither $k$ nor $m$ are related to the worst-case privacy considerations and should be selected based on the efficiency properties of the algorithm in reconstructing the signal from the noisy reports.

We ran a number of simulations (averaged over 10 replicates) to understand how these three parameters effect decoding; see Figure \ref{fig:sim}. All scenarios assumed $\epsilon = \ln(3)$ privacy guarantee. Since only a single report from each user was simulated, One-time \RAPPOR{} was used. Population sampled is shown in the first panel and contains 100 non-zero strings with 100 strings that had zero probability of occurring. Frequencies of non-zero strings followed an exponential distribution as shown in the figure.

\begin{figure}[t]
\centering
\includegraphics[width=\columnwidth]{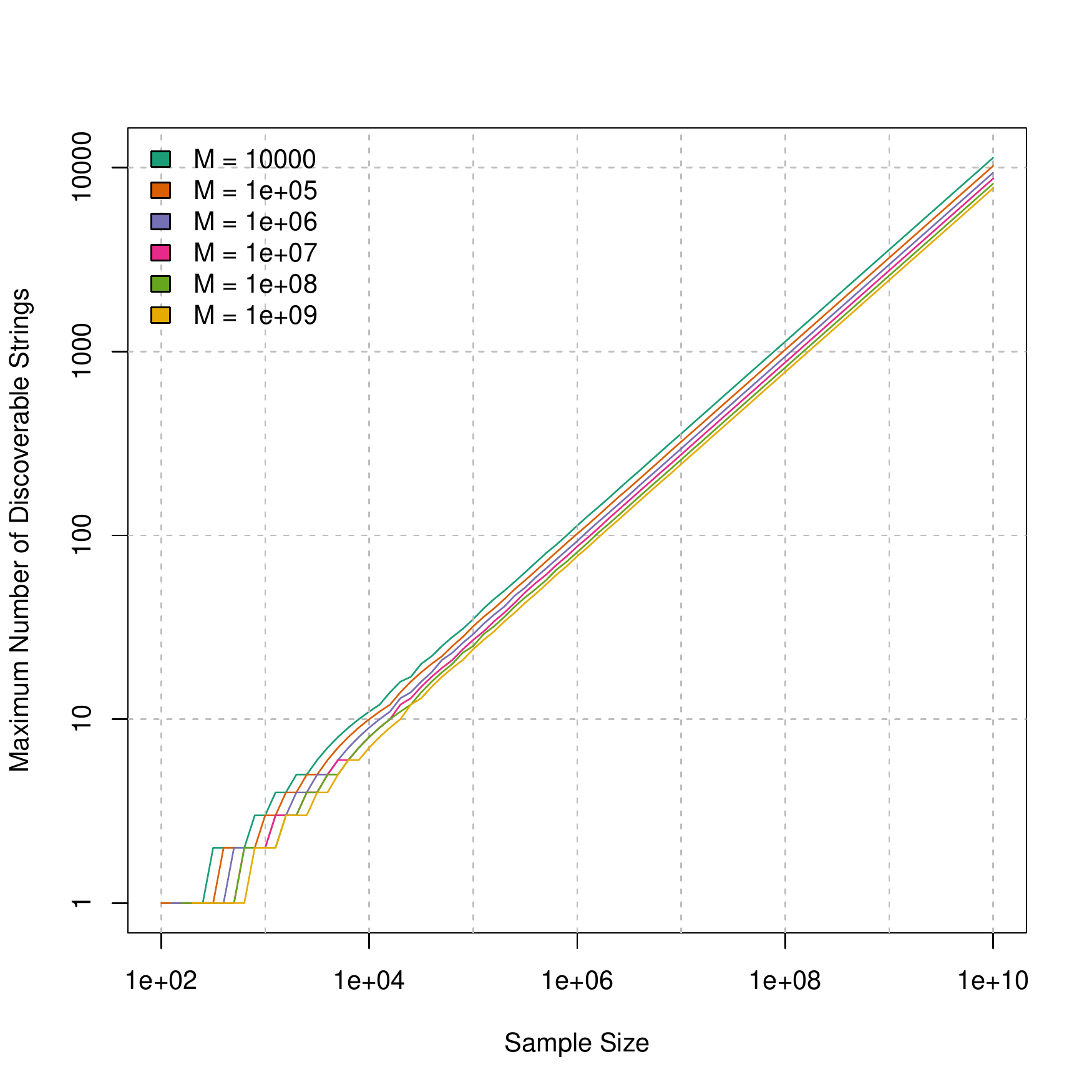}
\caption{Sample size vs the upper limit on the strings whose frequency can be learned. Seven colored lines represent different cardinalities of the candidate string set. Here, $p = 0.5$, $q = 0.75$ and $f = 0$.}
\label{fig:learn}
\end{figure}

In the other three panels, the x-axis shows the recall rate and the y-axis shows the precision rate. In all three panels, the same set of points are plotted and are only labeled differently depending on which parameter changes in a particular panel. Each point represents an average recall and precision for a unique combination of $k$, $h$, and $m$. For example, the second panel shows the effect of the Bloom filter size on both precision and recall while keeping both $h$ and $m$ fixed. It is difficult to make definitive conclusions about the optimal size of the Bloom filter as different sizes perform similarly depending on the values of $h$ and $m$. The third panel, however, shows a clear preference for using only two hash functions from the perspective of utility, as the decrease in the number of hash functions used increases the expected recall. The fourth panel, similarly to the second, does not definitively indicate the optimal direction for choosing the number of cohorts.

\subsection{What Can We Learn?}
In practice, it is common to use thresholds on the number of unique submissions in order to ensure some privacy. However, arguments as to how those thresholds should be set abound, and most
of the time they are based on a `feel' for what is accepted and lack any objective justification. \RAPPOR{} also requires $\epsilon$, a user-tunable parameter, which by the design of the algorithm translates into limits on frequency domain, i.e., puts a lower limit on the number of times a string needs to be observed in a sample before it can be reliably identified and its frequency estimated. Figure~\ref{fig:learn} shows the relationship between the sample size (x-axis) and the theoretical upper limit (y-axis) on how many strings can be detected at that sample size for a particular choice of $p = 0.5$ and $q = 0.75$ (with $f = 0$) at a given confidence level $\alpha = 0.05$.

It is perhaps surprising that we do not learn more at very large sample sizes (e.g., one billion). The main reason is that as the number of strings in the population becomes large, their frequencies proportionally decrease and they become hard to detect at those low frequencies.

We can only reliably detect about 10,000 strings in a sample of ten billion and about 1,000 with a sample of one hundred million. A general rule of thumb is $\sqrt{N}/10$, where $N$ is the sample size. These theoretical calculations are based on the Basic One-time \RAPPOR{} algorithm (the third modification) and are the upper limit on what can be learned since there is no additional uncertainty introduced by the use of Bloom filter. Details of the calculations are shown in the Appendix.

While providing $\ln(3)$-differential privacy for one time collection, if one would like to detect items with frequency 1\%, then one million samples are required,
0.1\% would require a sample size of 100 million and 0.01\% items would be identified only in a sample size of 10 billion. 

Efficiency of the unmodified \RAPPOR{} algorithm is significantly inferior when compared to the Basic One-time \RAPPOR{} (the price of compression). Even for the Basic One-time \RAPPOR{}, the provided bound can be theoretically achieved only if the underlying distribution of the strings' frequencies is uniform (a condition under which the smallest frequency is maximized). With the presence of several high-frequency strings, there is less probability mass left for the tail and, with the drop in their frequencies, their detectability suffers.

\section{Experiments and Evaluation}
We demonstrate our approach using two simulated and two real-world collection examples. The first simulated one uses the Basic One-time \RAPPOR{} where we learn the shape of the underlying Normal distribution. The second simulated example uses unmodified \RAPPOR{} to collect strings whose frequencies exhibit exponential decay. The third example is drawn from a real-world dataset on processes running on a set of Windows machines. The last example is based on the Chrome browser settings collections.

\subsection{Reporting on the Normal Distribution}
To get a sense of how effectively we can learn the underlying distribution of values reported through the Basic One-time \RAPPOR{}, we simulated learning the shape of the Normal distribution (rounded to integers) with mean 50 and standard deviation 10. The privacy constraints were: $q = 0.75$ and $p = 0.5$ providing $\epsilon = \ln(3)$ differential privacy ($f = 0$). Results are shown in Figure~\ref{fig:normal} for three different sample sizes. With 10,000 reports, results are just too noisy to obtain a good estimate of the shape. The Normal bell curve begins to emerge already with 100,000 reports and at one million reports it is traced very closely. Notice the noise in the left and right tails where there is essentially no signal. It is required by the differential privacy condition and also gives a sense of how uncertain our estimated counts are.

\begin{figure*}[!t]
\centering
\includegraphics[scale=.5]{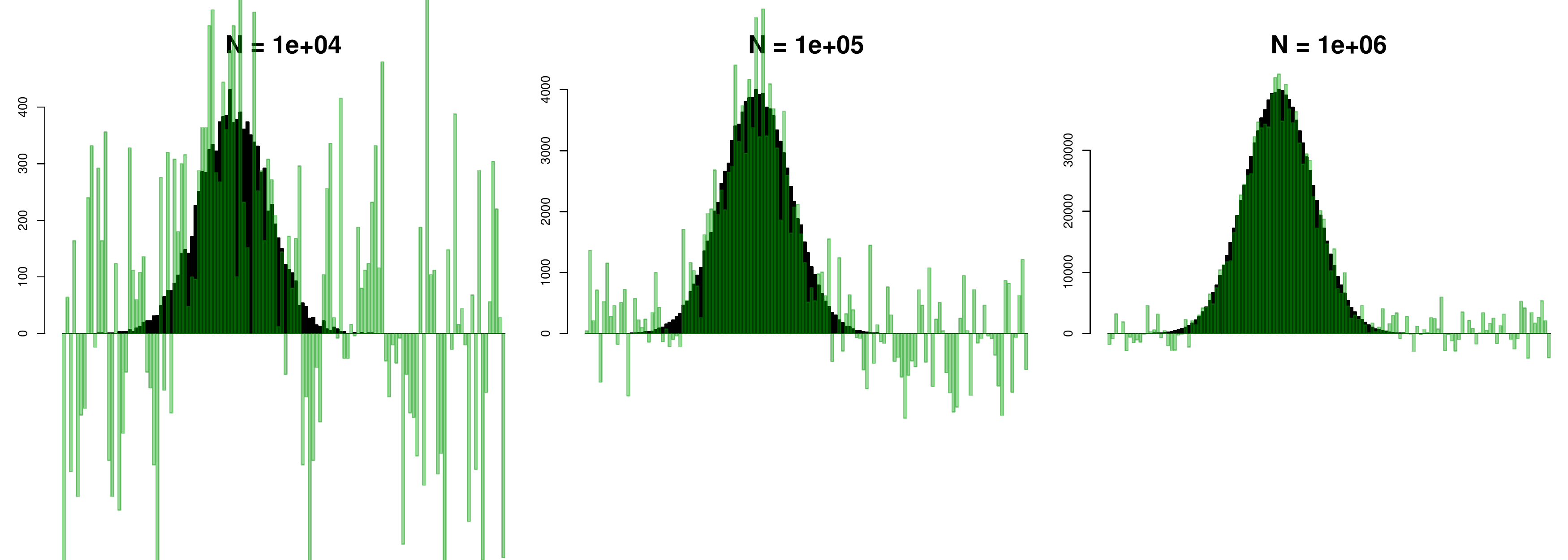}
\caption{Simulations of learning the normal distribution with mean 50 and standard deviation 10. The \RAPPOR{} privacy parameters are $q = 0.75$ and $p = 0.5$, corresponding to $\epsilon = \ln(3)$. True sample distribution is shown in black; light green shows the estimated distribution based on the decoded \RAPPOR{} reports. We do not assume \emph{a priori} knowledge of the Normal distribution in learning. If such prior information were available, we could significantly improve upon learning the shape of the distribution via smoothing.}
\label{fig:normal}
\end{figure*}

\subsection{Reporting on an Exponentially-distributed Set of Strings} 
The true underlying distribution of strings from which we sample is shown in Figure~\ref{fig:detected}. It shows commonly encountered exponential decay in the frequency of strings with several ``heavy hitters'' and the long tail. After sampling 1 million values (one collection event per user) from this population at random, we apply \RAPPOR{} to generate 1 million reports with $p = 0.5$, $q = 0.75$, $f = 0.5$, two hash functions, Bloom filter size of 128 bits and 16 cohorts.

After the statistical analysis using the Bonferroni correction discussed above, 47 strings were estimated to have counts significantly different from 0. Just 2 of the 47 strings were false positives, meaning their true counts were truly 0 but estimated to be significantly different. The top-20 detected strings with their count estimates, standard errors, p-values and z-scores (SNR) are shown in Table~\ref{tab:results20}. Small p-values show high confidence in our assessment that the true counts are much larger than 0 and, in fact, comparing columns 2 and 5 confirms that. Figure~\ref{fig:detected} shows all 47 detected strings in dark red. All common strings above the frequency of approximately $1\%$ were detected and the long tail remained protected by the privacy mechanism.

\begin{figure}[!t]
\centering
\includegraphics[trim=0 0.8in 0 0.6in,clip=true,width=\columnwidth]{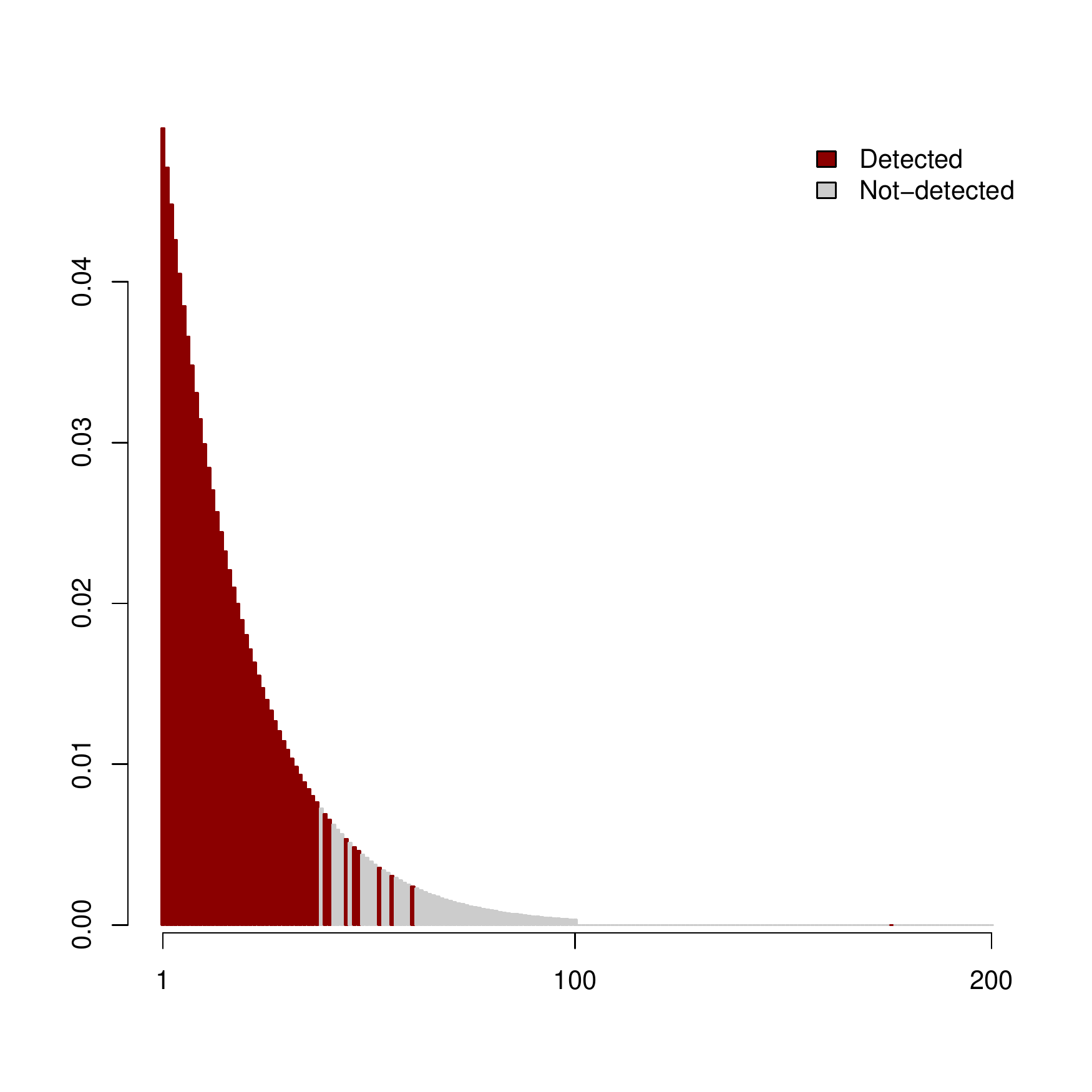}
\caption{Population of strings with their true frequencies on the vertical axis (0.01 is 1\%). Strings detected by \RAPPOR{} are shown in dark red.}
\label{fig:detected}
\end{figure}

\begin{table}[!h]
\centering
\begin{tabular}{lcccccr}
  \hline
String & Est.\ & Stdev\ & P.value & Truth & Prop.\ & SNR \\ 
  \hline
V\_1 & 48803 & 2808 & 5.65E-63 & 49884 & 0.05 & 17.38 \\ 
  V\_2 & 47388 & 2855 & 5.82E-58 & 47026 & 0.05 & 16.60 \\ 
  V\_5 & 41490 & 2801 & 4.30E-47 & 40077 & 0.04 & 14.81 \\ 
  V\_7 & 40682 & 2849 & 4.58E-44 & 36565 & 0.04 & 14.28 \\ 
  V\_4 & 40420 & 2811 & 1.31E-44 & 42747 & 0.04 & 14.38 \\ 
  V\_3 & 39509 & 2882 & 7.03E-41 & 44642 & 0.04 & 13.71 \\ 
  V\_8 & 36861 & 2842 & 5.93E-37 & 34895 & 0.03 & 12.97 \\ 
  V\_6 & 36220 & 2829 & 4.44E-36 & 38231 & 0.04 & 12.80 \\ 
  V\_10 & 34196 & 2828 & 1.72E-32 & 31234 & 0.03 & 12.09 \\ 
  V\_9 & 32207 & 2805 & 1.45E-29 & 33106 & 0.03 & 11.48 \\ 
  V\_12 & 30688 & 2822 & 9.07E-27 & 28295 & 0.03 & 10.87 \\ 
  V\_11 & 29630 & 2831 & 5.62E-25 & 29908 & 0.03 & 10.47 \\ 
  V\_14 & 27366 & 2850 & 2.33E-21 & 25984 & 0.03 & 9.60 \\ 
  V\_19 & 23860 & 2803 & 3.41E-17 & 20057 & 0.02 & 8.51 \\ 
  V\_13 & 22327 & 2826 & 4.69E-15 & 26913 & 0.03 & 7.90 \\ 
  V\_15 & 21752 & 2825 & 2.15E-14 & 24653 & 0.02 & 7.70 \\ 
  V\_20 & 20159 & 2821 & 1.26E-12 & 19110 & 0.02 & 7.15 \\ 
  V\_18 & 19521 & 2835 & 7.74E-12 & 20912 & 0.02 & 6.89 \\ 
  V\_17 & 18387 & 2811 & 7.86E-11 & 22141 & 0.02 & 6.54 \\ 
  V\_21 & 18267 & 2828 & 1.33E-10 & 17878 & 0.02 & 6.46 \\ 
   \hline
\end{tabular}
\caption{Top-20 strings with their estimated frequencies, standard deviations, p-values, true counts and signal to noise ratios (SNR or z-scores).}
\label{tab:results20}
\end{table}

\subsection{Reporting on Windows Process Names}\label{sec:windows}
We collected 186,792 reports from 10,133 different Windows computers, sampling actively running
processes on each machine. On average, just over 18 process names were collected from each machine with the goal of recovering the most common ones and estimating the frequency of a particularly malicious binary named ``BADAPPLE.COM''.

This collection used 128 Bloom filter with 2 hash functions and 8 cohorts. Privacy parameters were chosen such that $\epsilon_1 = 1.0743$ with $q = 0.75$, $p = 0.5,$ and $f = 0.5$. Given this configuration, we optimistically expected to discover processes with frequency of at least 1.5\%.

We identified 10 processes shown in Table \ref{tab:proc} ranging in frequency between 2.5\% and 4.5\%. They were identified by controlling the False Discovery Rate at 5\%. The ``BADAPPLE.COM'' process was estimated to have frequency of 2.6\%. The other 9 processes were common Windows tasks we would expect to be running on almost every Windows machine.

\begin{table}[h]
\centering
\caption{Windows processes detected.}
\begin{tabular}{lcccc}
  \hline
Process Name & Est.\ & Stdev & P.value & Prop.\ \\ 
  \hline
RASERVER.EXE & 8054 & 1212 & 1.56E-11 & 0.04    \\ 
  RUNDLL32.EXE & 7488 & 1212 & 3.32E-10 & 0.04  \\ 
  CONHOST.EXE & 7451 & 1212 & 4.02E-10 & 0.04   \\ 
  SPPSVC.EXE & 6363 & 1212 & 7.74E-08 & 0.03    \\ 
  AITAGENT.EXE & 5579 & 1212 & 2.11E-06 & 0.03  \\ 
  MSIEXEC.EXE & 5147 & 1212 & 1.10E-05 & 0.03   \\ 
  SILVERLIGHT.EXE & 4915 & 1212 & 2.53E-05 & 0.03  \\ 
  BADAPPLE.COM & 4860 & 1212 & 3.07E-05 & 0.03  \\ 
  LPREMOVE.EXE & 4787 & 1212 & 3.95E-05 & 0.03  \\ 
  DEFRAG.EXE & 4760 & 1212 & 4.34E-05 & 0.03    \\ 
   \hline
\end{tabular}
\label{tab:proc}
\end{table}

\subsection{Reporting on Chrome Homepages}\label{sec:chromehome}
The Chrome Web browser has implemented and deployed \RAPPOR{} 
to collect data about Chrome clients~\cite{ChromeRAPPORpage}.
Data collection has been limited to
some of the Chrome users who have opted in to send usage statistics to Google,
and to certain Chrome settings,
with daily collection 
from approximately $\sim$14 million respondents.

Chrome settings, such as homepage, search engine and others, are often targeted by malicious software and changed without users' consent. To understand who the main players are, it is critical to know the distribution of these settings on a large number of Chrome installations. Here, we focus on learning the distribution of homepages and demonstrate what can be learned from a dozen million reports with strong privacy guarantees.

This collection used 128 Bloom filter with 2 hash functions and 32 cohorts. Privacy parameters were chosen such that $\epsilon_1 = 0.5343$ with $q = 0.75$, $p = 0.5,$ and $f = 0.75$. 
Given this configuration, optimistically, \RAPPOR{} analysis can discover homepage URL domains, with statistical confidence, if their frequency exceeds 0.1\% of the responding population.
Practically, this means that more than $\sim$14 thousand clients must report on the same URL domain, before it can be identified in the population by \RAPPOR{} analysis.

Figure~\ref{fig:detectedtwo} shows the relative frequencies of 31 unexpected homepage domains discovered by \RAPPOR{} analysis.
(Since not all of these are necessarily malicious, the figure does not include the actual URL domain strings that were identified.)
As one might have expected, there are several popular homepages, likely intentionally set by users, along with a long tail of relatively rare URLs. Even though less than 0.5\% out of 8,616 candidate URLs provide enough statistical evidence for their presence (after the FDR correction), they collectively account for about 85\% of the total probability mass. 

\begin{figure}[!t]
\centering
\includegraphics[width=\columnwidth]{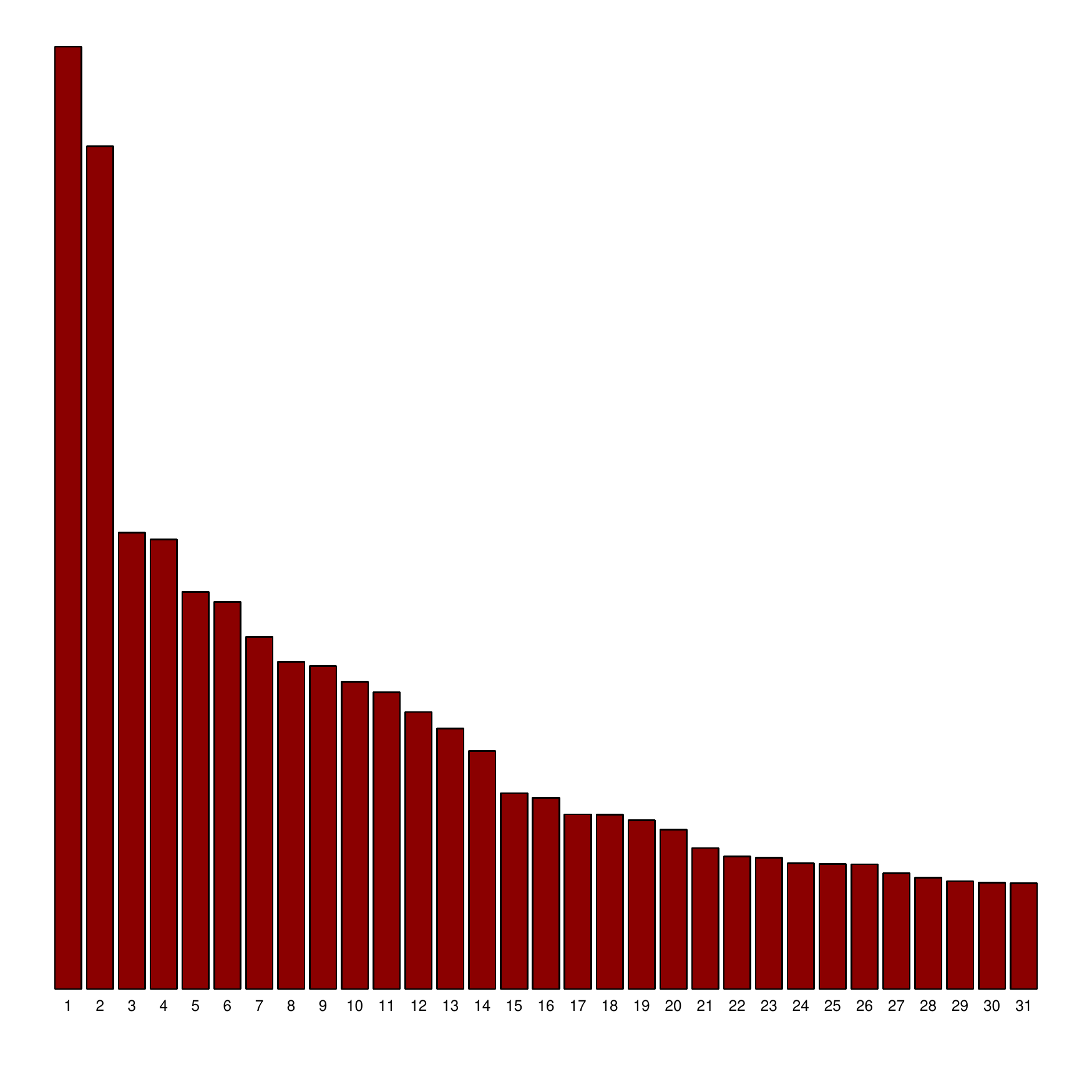}
\caption{Relative frequencies of the top 31 unexpected Chrome homepage domains found by analyzing $\sim$14 million \RAPPOR{} reports, excluding expected domains
(the homepage ``google.com'', etc.).}
\label{fig:detectedtwo}
\end{figure}

\section{Attack Models and Limitations}\label{sec:attacks}
We consider three types of attackers with different capabilities for collecting \RAPPOR{} reports. 

The least powerful attacker has access to a single report from each user and is limited by one-time differential privacy level $\epsilon_1$ on how much knowledge gain is possible. 
This attacker corresponds to an eavesdropper that has temporary ability to snoop on the users' reports.

A windowed attacker is presumed to have access to one client's data over a well-defined period of time. This attacker, depending on the sophistication of her learning model, could learn more information about a user than the attacker of the first type. Nevertheless, the improvement in her ability to violate privacy is strictly bounded by the longitudinal differential privacy guarantee of $\epsilon_{\infty}$. 
This more powerful attacker may correspond to an adversary such as a malicious Cloud service employee, who may have temporary access to reports, or access to a time-bounded log of reports.

The third type of attacker is assumed to have unlimited collection capabilities and can learn the Permanent randomized response $B'$ with absolute certainty. Because of the randomization performed to obtain $B'$ from $B$, she is also bounded by the privacy guarantee of $\epsilon_{\infty}$ and cannot improve upon this bound with more data collection.
This corresponds to a worst-case adversary, but still one that doesn't have direct access to the true data values on the client.

Despite envisioning a completely local privacy model, one where users themselves release data in a privacy-preserving fashion,
operators of \RAPPOR{} collections, however, can easily manipulate the process to learn more information than warranted by the
nominal $\epsilon_{\infty}$. Soliciting users to participate more than once in a particular collection results in multiple Permanent
randomized responses for each user and partially defeats the benefits of memoization. In the web-centric world, users use multiple accounts
and multiple devices and can unknowingly participate multiple times, releasing more information than what they expected.
This problem could be mitigated to some extent by running collections per account and sharing a common Permanent randomized response.
Notice the role of the operator to ensure that such processes are in place and the required or assumed trust on the part of the user. 

It is likely that some attackers will aim to target specific users 
by isolating and analyzing reports from that user, or a small group of users that includes them.
Even so, some randomly-chosen users need not fear such attacks at all:
with probability $\left(\frac{1}{2}f\right)^h$, clients will generate a Permanent randomized response $B'$ with all 0s at the positions of set Bloom filter bits. Since these clients are not contributing any useful information to the collection process, targeting them individually by an attacker is counter-productive. An attacker has nothing to learn about this particular user. 
Also, for all users, at all times, there is plausible deniability proportional to the fraction of clients providing no information.

\begin{figure}[t]
\begin{center}
\includegraphics[trim=0 0.2in 0 0,clip=true,width=\columnwidth]{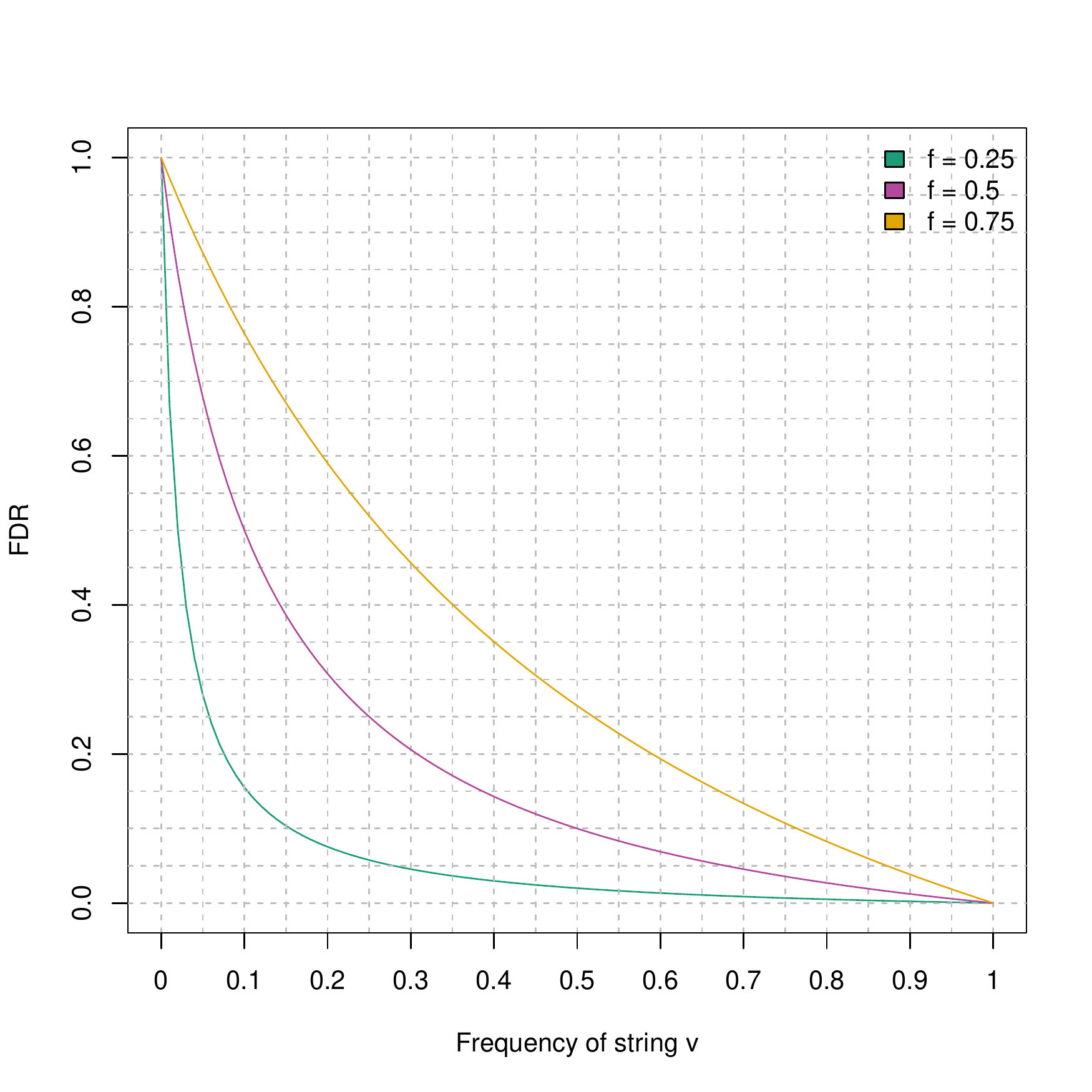}
\caption{False Discovery Rate (FDR) as a function of string frequency and $f$. Identifying rare strings in a population without introducing a large number of false discoveries is infeasible. Also, FDR is proportional to $f$.}
\label{fig:fdr}
\end{center}
\end{figure}

\begin{figure}[t]
\begin{center}
\includegraphics[trim=0 0.2in 0 0,clip=true,width=\columnwidth]{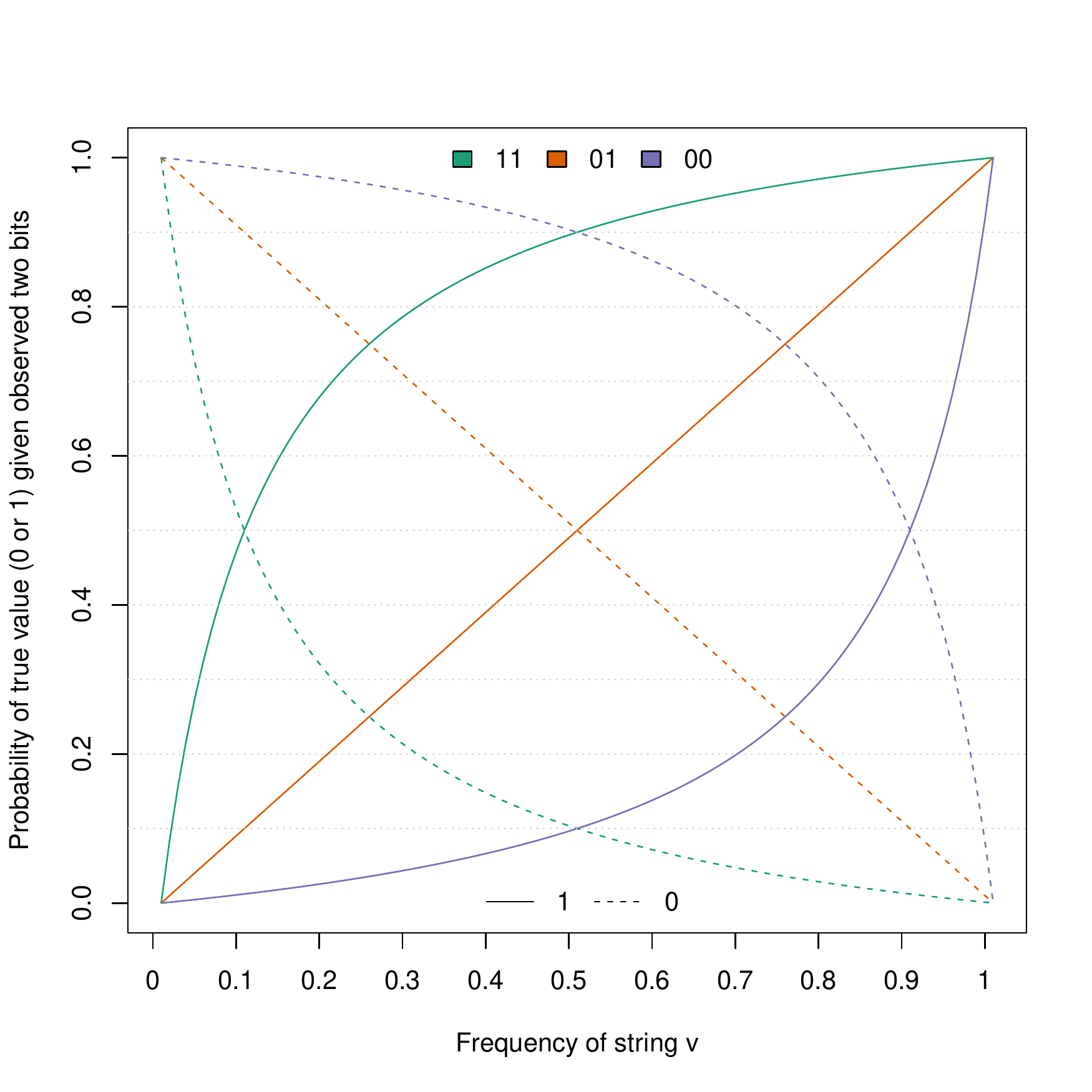}
\caption{Exact probabilities for inferring the true value $v$ given the two bits observed in a \RAPPOR{} report $S$ corresponding to the two bits set by string $v$. For rare strings, even when both bits are set to 1 (green lines), it is still much more likely that the client \emph{did not} report $v$, but some other value.}
\label{fig:fdr2}
\end{center}
\end{figure}

In one particular attack scenario, 
imagine an attacker that is interested in learning whether a given client has a particular value $v$,
whose population frequency is known to be $f_v$. 
The strongest evidence in support of $v$ comes in the form of both Bloom filter bits for $v$ being set in the client's report (if two hash functions are used). 
The attacker can formulate its target set by selecting all reports with these two bits set. 
However, this set will miss some clients with $v$ and include other clients who did not report $v$. 
False discovery rate (FDR) is the proportion of clients in the target set who reported a value different from $v$. 
Figure \ref{fig:fdr} shows FDR as a function of $f_v$, the frequency of the string $v$.
Notably,
for relatively rare values, most clients in the target set will, in fact, have a value that is different from $v$,
which will hopefully deter any would-be attackers.

The main reason for the high FDR rate at low frequencies $f_v$ stems from the limited evidence provided by the observed bits in support of $v$. 
This is clearly illustrated by 
Figure \ref{fig:fdr2} where the probability that $v$ was reported (1) or not reported (0) by the client is plotted as a function of $f_v$. 
For relatively rare strings (those with less than 10\% frequency), even when both bits corresponding to $v$ are set in the report, 
the probability of $v$ being reported is much smaller than of it not being reported. 
Because the prior probability $f_v$ is so small, a single client's reports cannot provide sufficient evidence in favor of $v$.

\subsection{Caution and Correlations}
Although it advances the state of the art, 
\RAPPOR{} is not a panacea, but rather simply a tool
that can provide significant benefits
when used cautiously, and correctly, 
using parameters appropriate to its application context.
Even then,
\RAPPOR{} should be used only 
as part of a comprehensive privacy-protection strategy,
which should include
limited data retention and
other pragmatic processes mentioned in Section~\ref{sec:motivation},
and already in use by Cloud operators.

As in previous work on differential privacy for database records, \RAPPOR{} provides privacy guarantees for the responses from individual clients. One of the limitations of our approach has to do with ``leakage'' of additional information when respondents use several clients that participate in the same collection event. In the real world, this problem is mitigated to some extent by intrinsic difficulty of linking different clients to the same participant. Similar issues occur when highly correlated, or even exactly the same, predicates are collected at the same time. This issue, however, can be mostly handled with careful collection design.

Such inadvertent correlations can arise in many different ways in \RAPPOR{} applications, in each case possibly leading to the collection of too much correlated information from a single client, or user, and a corresponding degradation of privacy guarantees.
Obviously, this may be more likely to happen if \RAPPOR{} reports are collected, from each client, on too many different client properties.
However, it may also happen in more subtle ways.
For example, the number of cohorts used in the collection design must be carefully selected and changed over time,
to avoid privacy implications;
otherwise, cohorts may be so small  as to facilitate the tracking of clients,
or clients may report as part of different cohorts over time,
which will reduce their privacy.
\RAPPOR{} responses can even affect client anonymity,
when they are collected on immutable client values that are the same across all clients:
if the responses contain too many bits (e.g., the Bloom filters are too large),
this can facilitate tracking clients,
since the bits of the Permanent randomized responses are correlated.
Some of these concerns may not apply in practice
(e.g., tracking responses may be infeasible, because of encryption),
but all must be considered in \RAPPOR{} collection design.

In particular, longitudinal privacy protection guaranteed by the Permanent randomized response assumes that client's value does not change over time. It is only slightly violated if the value changes very slowly. In a case of rapidly changing, correlated stream of values from a single user, additional measures must be taken to guarantee longitudinal privacy. The practical way to implement this would be to budget $\epsilon_{\infty}$ over time, spending a small portion on each report. In the \RAPPOR{} algorithm this would be equivalent to letting $q$ get closer and closer to $p$ with each collection event.

Because differential privacy deals with the worst-case scenario, the uncertainty introduced by the Bloom filter does not play any role in the calculation of its bounds. Depending on the random draw, there may or may not be multiple candidate strings mapping to the same $h$ bits in the Bloom filter. For the average-case privacy analysis, however, Bloom filter does provide additional privacy protection (a flavor of $k$-anonymity) because of the difficulty in reliably inferring a client's value $v$ from its Bloom filter representation $B$ \citep{bloom_privacy}.

\section{Related Work}
Data collection from clients in a way that preserves their privacy and at the same time enables meaningful aggregate inferences is an active area of research both in academia and industry.
Our work fits into the category of problems recently explored by \cite{Hsu2012, SafeZones, Chan2012, Liu2012}, where an untrusted aggregator wishes to learn the ``heavy hitters" in the clients' data or run certain types of learning algorithms on the aggregated data, while guaranteeing the privacy of each contributing client and, in some cases, restricting the size of the communication from the client to the untrusted aggregator. Our contribution is to suggest an alternative to those already explored that is intuitive, easy-to-implement, and potentially more suitable to certain learning problems, and to provide a detailed statistical decoding methodology for our approach, as well as experimental data on its performance. Furthermore, in addition to guaranteeing differential privacy, we make explicit algorithmic steps towards protection against linkability across reports from the same user.

It is natural to ask why we built our mechanisms upon randomized response, rather than upon two primitives most commonly used to achieve differential privacy: the Laplace and Exponential mechanisms~\cite{dwork06,McSherryT07}. 
The Laplace mechanism is not suitable because the client's reported values may be categorical, rather than numeric, in which case direct noise addition does not make semantic sense. The Exponential mechanism is not applicable due to our desire to implement the system in a local model, where the privacy is ensured by each client individually without a need for a trusted third party. In that case, the client does not have sufficient information about the data space in order to do the necessary biased sampling required by the Exponential mechanism. Finally, randomized response has the additional benefit of being relatively easy to explain to the end user, making the reasoning about the algorithm used to ensure privacy more accessible than other mechanisms implementing differential privacy.

Usage of various dimensionality reduction techniques in order to improve the privacy properties of algorithms while retaining utility is also fairly common~\cite{Liu2012, JL, AggarwalY07, MirMNW11}. Although our reliance on Bloom filters is driven by a desire to obtain a compact representation of the data in order to lower each client's potential transmission costs and the desire to use technologies that are already widely adopted in practice~\cite{BroderM03}, the related work in this space with regards to privacy~\cite{bloom_privacy} may be a source for optimism as well. 
It is conceivable that through a careful selection of hash functions, 
or choice of other Bloom filter parameters, 
it may be possible to further raise privacy defenses against attackers,
although we have not explored that direction in much detail.

The work most similar to ours is by Mishra and Sandler~\cite{MishraS06}. One of the main additional contributions of our work is the more extensive decoding step, that provides both experimental and statistical analyses of collected data for queries that are more complex than those considered in their work. The second distinction is our use of the second randomization step, the Instantaneous randomized response, in order to make the task of linking reports from a single user difficult, along with more detailed models of attackers' capabilities.

The work of~\cite{OurData} approaches the challenge of eliminating the need for a trusted aggregator with a distributed solution that places trust in other clients instead. \cite{chen2012towards} and \cite{Akkus2012} implement a differentially private protocol over distributed user data by relying on an honest-but-curious proxy or data aggregator bound by certain commitments.

Several lines of work aim to address the question of longitudinal data collection with privacy. The work of \cite{MedianMech} considers scenarios when many predicate queries are asked against the same dataset, and it uses an approach that, rather than providing randomization for each answer separately, attempts to reconstruct the answer to some queries based on the answers previously given to other queries. The high-level idea of \RAPPOR{} bears some resemblance to this technique--the Instantaneous randomized response is reusing the result of the Permanent randomized response step. However, the overall goal is different---rather than answering a diverse number of queries, \RAPPOR{} collects reports to the same query over data that may be changing over time. 
Although it does not operate under the same local model as \RAPPOR{}, recent work by \cite{DworkNPRY10} on pan-private streaming and by~\cite{Dwork2010} on privacy under continual observation introduces additional ideas relevant for the longitudinal data collection with privacy.
 







\section{Summary}
\RAPPOR{} is a flexible, mathematically rigorous and practical platform for anonymous data collection
for the purposes of privacy-preserving crowdsourcing of population statistics on client-side data.
\RAPPOR{} gracefully handles multiple data collections from the same client by providing well-defined longitudinal differential privacy guarantees. Highly tunable parameters allow to balance risk versus utility over time, depending on one's needs and assessment of likelihood of different attack models. \RAPPOR{} is purely a client-based privacy solution. It eliminates the need for a trusted third-party server and puts control over client's data back into their own hands.

\paragraph{Acknowledgements.}
The authors would like to thank our many colleagues at Google and its Chrome team who have helped with this work, with special thanks due to Steve Holte and Moti Yung.
Thanks also to the CCS reviewers, and many others who have provided insightful feedback on the ideas, and this paper, 
in particular, Frank McSherry, Arvind Narayanan, Elaine Shi, and Adam D.\ Smith.

\bibliographystyle{plainnat}
\bibliography{submission}

\begin{thebibliography}{10}

\bibitem{AggarwalY07}
Charu~C. Aggarwal and Philip~S. Yu.
\newblock On privacy-preservation of text and sparse binary data with sketches.
\newblock In {\em Proceedings of the 2007 SIAM International Conference on Data
  Mining (SDM)}, pages 57--67, 2007.

\bibitem{Akkus2012}
Istemi~Ekin Akkus, Ruichuan Chen, Michaela Hardt, Paul Francis, and Johannes
  Gehrke.
\newblock Non-tracking web analytics.
\newblock In {\em Proceedings of the 2012 ACM Conference on Computer and
  Communications Security (CCS)}, pages 687--698, 2012.

\bibitem{Benjamini1995}
Yoav Benjamini and Yosef Hochberg.
\newblock Controlling the false discovery rate: A practical and powerful
  approach to multiple testing.
\newblock {\em Journal of the Royal Statistical Society Series B
  (Methodological)}, 57(1):289--300, 1995.

\bibitem{bloom_privacy}
Giuseppe Bianchi, Lorenzo Bracciale, and Pierpaolo Loreti.
\newblock `{Better Than Nothing}' privacy with {Bloom} filters: To what extent?
\newblock In {\em Proceedings of the 2012 International Conference on Privacy
  in Statistical Databases (PSD)}, pages 348--363, 2012.

\bibitem{bloom}
Burton~H. Bloom.
\newblock Space/time trade-offs in hash coding with allowable errors.
\newblock {\em Commun. ACM}, 13(7):422--426, July 1970.

\bibitem{BroderM03}
Andrei~Z. Broder and Michael Mitzenmacher.
\newblock Network applications of {B}loom filters: A {S}urvey.
\newblock {\em Internet Mathematics}, 1(4):485--509, 2003.

\bibitem{Chan2012}
T.-H.~Hubert Chan, Mingfei Li, Elaine Shi, and Wenchang Xu.
\newblock Differentially private continual monitoring of heavy hitters from
  distributed streams.
\newblock In {\em Proceedings of the 12th International Conference on Privacy
  Enhancing Technologies (PETS)}, pages 140--159, 2012.

\bibitem{chen2012towards}
Ruichuan Chen, Alexey Reznichenko, Paul Francis, and Johannes Gehrke.
\newblock Towards statistical queries over distributed private user data.
\newblock In {\em Proceedings of the 9th USENIX Conference on Networked Systems
  Design and Implementation (NSDI)}, pages 169--182, 2012.

\bibitem{ChromeRAPPORpage}
Chromium.org.
\newblock {D}esign {D}ocuments: {RAPPOR} ({R}andomized {A}ggregatable {P}rivacy
  {P}reserving {O}rdinal {R}esponses).
\newblock \url{http://www.chromium.org/developers/design-documents/rappor}.

\bibitem{DworkCACM}
Cynthia Dwork.
\newblock A firm foundation for private data analysis.
\newblock {\em Commun. ACM}, 54(1):86--95, January 2011.

\bibitem{OurData}
Cynthia Dwork, Krishnaram Kenthapadi, Frank McSherry, Ilya Mironov, and Moni
  Naor.
\newblock Our data, ourselves: Privacy via distributed noise generation.
\newblock In {\em Proceedings of 25th Annual International Conference on the
  Theory and Applications of Cryptographic Techniques (EUROCRYPT)}, pages
  486--503, 2006.

\bibitem{dwork06}
Cynthia Dwork, Frank Mc{S}herry, Kobbi Nissim, and Adam Smith.
\newblock Calibrating noise to sensitivity in private data analysis.
\newblock In {\em Proceedings of the 3rd Theory of Cryptography Conference
  (TCC)}, pages 265--284, 2006.

\bibitem{Dwork2010}
Cynthia Dwork, Moni Naor, Toniann Pitassi, and Guy~N. Rothblum.
\newblock Differential privacy under continual observation.
\newblock In {\em Proceedings of the 42nd ACM Symposium on Theory of Computing
  (STOC)}, pages 715--724, 2010.

\bibitem{DworkNPRY10}
Cynthia Dwork, Moni Naor, Toniann Pitassi, Guy~N. Rothblum, and Sergey
  Yekhanin.
\newblock Pan-private streaming algorithms.
\newblock In {\em Proceedings of The 1st Symposium on Innovations in Computer
  Science (ICS)}, pages 66--80, 2010.

\bibitem{HsuGHKNPR14}
Justin Hsu, Marco Gaboardi, Andreas Haeberlen, Sanjeev Khanna, Arjun Narayan,
  Benjamin~C. Pierce, and Aaron Roth.
\newblock Differential privacy: An economic method for choosing epsilon.
\newblock In {\em Proceedings of 27th IEEE Computer Security Foundations
  Symposium (CSF)}, 2014.

\bibitem{Hsu2012}
Justin Hsu, Sanjeev Khanna, and Aaron Roth.
\newblock Distributed private heavy hitters.
\newblock In {\em Proceedings of the 39th International Colloquium Conference
  on Automata, Languages, and Programming (ICALP) - Volume Part I}, pages
  461--472, 2012.

\bibitem{JL}
Krishnaram Kenthapadi, Aleksandra Korolova, Ilya Mironov, and Nina Mishra.
\newblock Privacy via the {J}ohnson-{L}indenstrauss transform.
\newblock {\em Journal of Privacy and Confidentiality}, 5(1):39--71, 2013.

\bibitem{SafeZones}
Daniel Keren, Guy Sagy, Amir Abboud, David Ben-David, Assaf Schuster, Izchak
  Sharfman, and Antonios Deligiannakis.
\newblock Monitoring distributed, heterogeneous data streams: The emergence of
  safe zones.
\newblock In {\em Proceedings of the 1st International Conference on Applied
  Algorithms (ICAA)}, pages 17--28, 2014.

\bibitem{KiferM11}
Daniel Kifer and Ashwin Machanavajjhala.
\newblock No free lunch in data privacy.
\newblock In {\em Proceedings of the ACM SIGMOD International Conference on
  Management of Data (SIGMOD)}, pages 193--204, 2011.

\bibitem{Liu2012}
Bin Liu, Yurong Jiang, Fei Sha, and Ramesh Govindan.
\newblock Cloud-enabled privacy-preserving collaborative learning for mobile
  sensing.
\newblock In {\em Proceedings of the 10th ACM Conference on Embedded Network
  Sensor Systems (SenSys)}, pages 57--70, 2012.

\bibitem{McSherryT07}
Frank Mc{S}herry and Kunal Talwar.
\newblock Mechanism design via differential privacy.
\newblock In {\em Proceedings of the 48th Annual IEEE Symposium on Foundations
  of Computer Science (FOCS)}, pages 94--103, 2007.

\bibitem{MirMNW11}
Darakhshan~J. Mir, S.~Muthukrishnan, Aleksandar Nikolov, and Rebecca~N. Wright.
\newblock Pan-private algorithms via statistics on sketches.
\newblock In {\em Proceedings of Symposium on Principles of Database Systems
  (PODS)}, pages 37--48, 2011.

\bibitem{mironov-CCS12}
Ilya Mironov.
\newblock On significance of the least significant bits for differential
  privacy.
\newblock In {\em Proceedings of ACM Conference on Computer and Communications
  Security (CCS)}, pages 650--661, 2012.

\bibitem{MishraS06}
Nina Mishra and Mark Sandler.
\newblock Privacy via pseudorandom sketches.
\newblock In {\em Proceedings of Symposium on Principles of Database Systems
  (PODS)}, pages 143--152, 2006.

\bibitem{MedianMech}
Aaron Roth and Tim Roughgarden.
\newblock Interactive privacy via the median mechanism.
\newblock In {\em Proceedings of the 42nd ACM Symposium on Theory of Computing
  (STOC)}, pages 765--774, 2010.

\bibitem{lasso}
Robert Tibshirani.
\newblock Regression shrinkage and selection via the {Lasso}.
\newblock {\em Journal of the Royal Statistical Society, Series B},
  58:267--288, 1994.

\bibitem{warner}
Stanley~L. Warner.
\newblock Randomized response: A survey technique for eliminating evasive
  answer bias.
\newblock {\em Journal of the American Statistical Association}, 60(309):pp.
  63--69, 1965.

\bibitem{WikipediaRR}
Wikipedia.
\newblock Randomized response.
\newblock \url{http://en.wikipedia.org/wiki/Randomized_response}.

\end{thebibliography}


\section*{Appendix}
\begin{obs}\label{obs-ratios}
For $a, b \geq 0$ and $c, d > 0: \frac{a+b}{c+d} \leq \max(\frac{a}{c}, \frac{b}{d})$.
\end{obs}
\begin{proof} 
Assume wlog that $\frac{a}{c} \geq \frac{b}{d}$, and suppose the statement is false, i.e., $\frac{a+b}{c+d} > \frac{a}{c}$. Then $ac + bc > ac + ad$ or $bc > ad$, a contradiction with assumption that $\frac{a}{c} \geq \frac{b}{d}$.
\end{proof}

\subsection*{Deriving Limits on Learning}
We consider a Basic One-time \RAPPOR{} algorithm to establish theoretical limits on what can be learned using a particular parameter configuration and a number of collected reports $N$. Since the Basic One-time  \RAPPOR{} is more efficient (lossless) than the original \RAPPOR{}, the following provides a strict upper bound for all \RAPPOR{} modifications.

\vfill\eject

Decoding for the Basic \RAPPOR{} is quite simple. Here, we assume that $f = 0$. The expected number that bit $i$ is set in a set of reports, $C_i$, is given by
$$
E(C_i) = qT_i + p(N - T_i),
$$
where $T_i$ is the number of times bit $i$ was truly set (was the signal bit). This immediately provides the estimator
$$
\hat{T}_i = \frac{C_i - pN}{q - p}.
$$

It can be shown that the variance of our estimator under the assumption that $T_i = 0$ is given by
$$
Var(\hat{T}_i) = \frac{p(1 - p)N}{(q - p)^2}.
$$

Determining whether $T_i$ is larger than 0 comes down to statistical hypothesis testing with $H_0: T_i = 0$ vs $H_1: T_i > 0$. Under the null hypothesis $H_0$ and letting $p = 0.5$, the standard deviation of $T_i$ equals
$$
\text{sd}(\hat{T}_i) = \frac{\sqrt{N}}{2q - 1}.
$$

We reject $H_0$ when
\begin{eqnarray*}
\hat{T}_i & > & Q \times \text{sd}(\hat{T}_i) \\
          & > & \frac{Q\sqrt{N}}{2q - 1},
\end{eqnarray*}
where $Q$ is the critical value from the standard normal distribution $Q = \Phi^{-1}(1 - \frac{0.05}{M})$ ($\Phi^{-1}$ is the inverse of the standard Normal cdf). Here, $M$ is the number of tests; in this case, it is equal to $k$, the length of the bit array. Dividing by $M$, the Bonferroni correction, is necessary to adjust for multiple testing to avoid a large number of false positive findings.

Let $x$ be the largest number of bits for which this condition is true (i.e., rejecting the null hypothesis). $x$ is maximized when $x$ out of $M$ items have a uniform distribution and a combined probability mass of almost 1. The other $M - x$ bits have essentially 0 probability. In this case, each non-zero bit will have frequency $1 / x$ and its expected count will be $E(\hat{T}_i) = N / x \text{ }\forall i$.

Thus we require
$$
\frac{N}{x} > \frac{Q\sqrt{N}}{2q - 1},
$$
where solving for $x$ gives
$$
x \le \frac{(2q - 1)\sqrt{N}}{Q}.
$$

\end{document}